\documentclass[12pt,titlepage]{article}
\usepackage{amsmath,amsthm,amssymb}
\usepackage[dvips]{graphicx}
\usepackage{array}
\usepackage{booktabs,arydshln}
\usepackage[mathscr]{eucal}
\usepackage{eurosym}
\usepackage{subfigure}
\usepackage{color}
\usepackage{listings}
\usepackage{multirow} 
\usepackage{algorithm}
\usepackage{algpseudocode}
\usepackage{xcolor}
\definecolor{darkgreen}{rgb}{0,0.5,0}

\newtheorem{Theorem}{Theorem}[section]

\newtheorem{Proposition}[Theorem]{Proposition}
\newtheorem{Property}[Theorem]{Property}

\theoremstyle{definition}
\newtheorem{Definition}[Theorem]{Definition}
\newtheorem{Example}[Theorem]{Example}
\theoremstyle{remark}
\newtheorem{Remark}[Theorem]{Remark}
\numberwithin{equation}{section}
\numberwithin{figure}{section}
\numberwithin{table}{section}

\newcommand{\argmin}{\operatornamewithlimits{argmin}}

\newcommand{\Var}{\mathrm{Var}}
\newcommand{\Esp}{\mathrm{E}}

\newcommand{\lcx}{\preceq_{\mathrm{cx}}}

\newcommand{\bi}{ \begin{itemize}  }
\newcommand{\ei}{\end{itemize}}

\newcommand{\bfor}{ \begin{eqnarray*} }
\newcommand{\efor}{\end{eqnarray*}}

\newcommand{\Xvec}{\boldsymbol{X}}

\newcommand{\xvec}{\boldsymbol{x}}

\begin{document}

\title{BALANCE AND FAIRNESS THROUGH MULTICALIBRATION IN NONLIFE INSURANCE PRICING}
\author{
Michel Denuit\\
Institute of Statistics, Biostatistics and Actuarial Science\\
Louvain Institute of Data Analysis and Modeling\\
UCLouvain\\
Louvain-la-Neuve, Belgium
\\[2mm]
Marie Michaelides\\
School of Mathematical and Computer Sciences\\
Actuarial Mathematics and Statistics\\
Heriot-Watt University\\
Edinburgh, United Kingdom
\\[2mm]
Julien Trufin\\
Department of Mathematics\\
Universit\'e Libre de Bruxelles (ULB)\\
Brussels, Belgium
}

\maketitle

\begin{abstract}
Autocalibration is known to be an important requirement for insurance premiums since it guarantees
that premium income balances corresponding claims, on average, not only at portfolio level but also inside each group paying similar premiums. Also, fairness has become a major concern because unfair treatment may expose insurers to lawsuits or reputational damage. Translating fairness into conditional mean independence allows actuaries to combine autocalibration and fairness into
the multicalibration concept. This paper studies the properties of multicalibration in an insurance
context and proposes practical ways to implement it, through local regression or bias correction within
groups including credibility adjustments. A case study based on motor insurance data illustrates the relevance of multicalibration in insurance pricing.
\\[3mm]
\end{abstract}

\section{Introduction and motivation}

Financial equilibrium is an important property in insurance pricing as it ensures that premiums
match corresponding claims, on average. This property must hold at portfolio level (global
balance) but also within meaningful groups of policyholders.
Numerical evidence suggests that candidate premiums produced by machine learning tools, like boosted regression
trees or neural networks, can depart substantially from observed claim totals. See, e.g., Denuit et al. (2019) and the references
therein. Autocalibration has been proposed as a remedy by Denuit et al. (2021), Denuit and Trufin (2023, 2024),
Lindholm et al. (2023), W\"uthrich (2023) and W\"uthrich and Ziegel (2024) since it ensures that
the amounts collected from policyholders paying the same premium match the
corresponding claim totals, on average.

Besides financial equilibrium, discrimination is another important issue in the private insurance industry.
We refer the reader, e.g., to Frees and Huang (2023) and Charpentier (2024) for a general overview of the topic.
The ban of gender-based discrimination within the European Union (EU)
is emblematic in that respect. In 2011, the European Court of Justice 
concluded that any gender-based insurance discrimination must be prohibited.
From December 21, 2012, all insurers operating in the EU are required to offer unisex premiums and benefits. 
In North-America, race is another example of protected feature in several states and provinces.
The protected attributes listed in the EU Directives are particularly relevant for insurers operating
in the EU. In addition to gender, this list includes some other features that are often used in insurance
like age, health or disability status, for instance. As another example, Fr\"ohlich and Williamson (2024) consider
``having migrant background'' as sensitive feature in their study of fair machine learning techniques.

The EU directive banning the use of gender in insurance risk classification imposes individual fairness.
In general, this means that the amounts of premium charged to policyholders differing only
on the prohibited feature must be equal. In this case, the sensitive feature
is not included in the risk classification scheme and commercial premiums do not depend on the
sensitive feature in an explicit way. However, indirect, or proxy discrimination may remain present
when some rating factors included in the price list are correlated to the omitted sensitive feature.
For instance, motor insurance premiums often differ according to the power of the vehicle or the distance traveled. If these features are correlated to gender, because men tend to drive more powerful cars over longer distances for example, then insurance premiums still indirectly discriminate according to gender. 
EU guidelines precisely consider this case, as explained later on.

This paper considers the case of a sensitive feature not subject to individual fairness by the law.
The insurer may thus let commercial premiums vary according to the sensitive feature. Typical examples
include protected variables like health or disability status that are used by insurers or proxy variables
to a forbidden rating factor. Because the use of these features may expose the insurer to complaints
from consumers organizations or formal investigation by the regulatory authorities, it
is important to be able to prove that the price list is fair with respect to the
sensitive feature under scrutiny, in some sense to be defined.

There are several approaches to fairness in the insurance and in the machine learning literature.
Let us mention actuarial fairness, individual fairness (or fairness through unawareness after
Dwork et al., 2012), counterfactual fairness (Kusner et al., 2017), and
several group-fairness criteria (independence, separation, and sufficiency after Barocas et al., 2019).
In this paper, we follow Baumann and Loi (2023) who advocate the relevance of a weak version of sufficiency
in insurance pricing. This group-fairness notion imposes conditional mean independence of the response and the sensitive feature, given the premium. Being based on conditional mean independence, sufficiency
is intuitively appealing and nicely combines with autocalibration, resulting in the concept
of multicalibration introduced by Hebert-Johnson et al. (2018). In words, a premium is
multicalibrated with respect to a collection of groups partitioning the portfolio if it
is autocalibrated within each group, in addition to being autocalibrated overall. The groups
considered in this paper are formed based on the sensitive feature.
The reinforcement of sufficiency into multicalibration is mentioned
as an avenue for future research in Section 3.3 in Baumann and Loi (2023) while calibration
conditional on gender and on salary has been implemented by Fissler et al. (2022)
in their application to workers' compensation insurance. This paper develops
these ideas and bridges fairness to autocalibration.

The remainder of this paper is organized as follows. The insurance ratemaking problem
is formalized in Section 2. Section 3 recalls the autocalibration concept.
Section 4 is devoted to fairness notions. The group-fairness concept based on conditional
mean independence considered in this paper is introduced there. Section 5 unites autocalibration
fairness through multicalibration while Section 6 explains how to implement multicalibration
through multibalance correction. Numerical illustrations are given in Section 7.
The final Section 8 summarizes the main contribution of this paper and discusses the results.

\section{Nonlife insurance pricing}

The machine learning literature mainly considers binary outcomes and algorithmic decision-making systems.
Typical examples include criminal risk assessment, hiring, college admission, lending or social services
intervention, for instance, corresponding to ``yes-no'' situations. A binary response is useful at underwriting stage in insurance where the insurer decides to cover the risk or to refuse it.
However, claim-related responses (claim numbers, claim severities, and claims totals) enter the
analysis when we consider insurance pricing and these responses are not binary (integer, continuous or mixed nature).

In this paper, we consider a general claim-related response $Y\geq 0$.
To fix the ideas, it is useful to consider that $Y$ corresponds to the annual claim frequency
for a policyholder in the portfolio. In accordance with credibility theory, we denote as $M$
the true mean value of $Y$, such that
$$
\Esp[Y|M=\mu]=\mu.
$$
The random variable $M$ is assumed to have finite expectation. Notice that the last identity can
be rewritten as $\Esp[Y|M]=M$, almost surely. Throughout the paper, every equality between
random variables is assumed to hold with probability 1, or almost surely.

The true mean $M$ remains unobserved and purely conceptual but is nevertheless useful to formalize
the insurance ratemaking problem where $M$ corresponds to the pure premium that should
be charged by the insurer. This is because individual departures $Y-M$ cancel out in any large
insurance portfolio where a law of large numbers applies. Charging $M$ corresponds to actuarial
fairness, which is practically impossible because $M$ cannot be observed. Contrary to algorithmic
fairness problems, risk classification in insurance targets $M$ which cannot be observed at
individual level. At the collective level, the distribution of $M$ reflects the composition of the
insured population targeted by the insurance company under consideration.

Since $M$ is unknown, the insurance company can only use information recorded in its database
about each policyholder. This corresponds to the set of features $X_1,\ldots,X_p$ gathered in the 
random vector $\Xvec$. Here, the information $\Xvec$ is used as a substitute to $M$ so that we assume
that $Y$ and $\Xvec$ are independent given $M$.
Actuarial pricing then aims to evaluate the pure premium as accurately as possible, based on $\Xvec$.
This means that
the target is the conditional expectation $\mu(\Xvec)=\Esp[Y|\Xvec]$ of the response $Y$
given the available information $\Xvec$.
Henceforth, $\mu(\Xvec)$ is referred to as the best-estimate price after Lindholm et al. (2022).
Notice that the conditional independence assumption allows us to write
$$
\mu(\Xvec)=\Esp\big[\Esp[Y|\Xvec,M]\big|\Xvec\big]=\Esp[M|\Xvec]
$$
so that $\mu(\Xvec)$ also approximates the individual premium $M$ based on the information available
to the insurance company.

The unknown function $\xvec\mapsto\mu(\xvec)=\Esp[Y|\Xvec=\xvec]$ is
approximated by a technical premium $\xvec\mapsto \pi_0(\xvec)$.
The technical premium is the most accurate approximation to $\mu(\Xvec)$, for internal use, only.
It allows the insurance company to compute policyholder value, for instance, by comparing
the anticipated cost $\pi_0(\Xvec)$ to the actual revenues generated by the contract.
These revenues correspond to the commercial premium $\pi(\Xvec)$ actually charged
for coverage (excluding taxes, commissions and general expenses). The premium structure $\pi(\cdot)$
is generally simplified compared to $\pi_0(\cdot)$, dropping some features or using them at a less
granular level.

Henceforth, we assume that $\pi_0(\Xvec)$, $\pi(\Xvec)$, and $\mu(\Xvec)$ are continuous random variables
admitting positive probability density functions over $(0,\infty)$.
This eases exposition and is generally the case when there is at least one continuous feature contained in the
available information $\Xvec$ and the functions $\pi_0(\cdot)$ and $\pi(\cdot)$ are continuously increasing functions of real scores built from $\Xvec$. Typically, a log-link is used and
$\pi_0(\xvec)=\exp(\mathrm{score}_0(\xvec))$
where $\mathrm{score}_0(\xvec)$ is obtained from boosting, random forests or neural networks. Also,
$\pi(\xvec)$ is the exponential of an additive score produced by a Generalized Additive Model for transparency
and to comply with the multiplicative premium structure implemented in the majority of insurance
commercial IT systems. 
Notice that we assume throughout the paper that commercial premiums vary
according to rating factors related to risk, only, without non-risk based corrections
(like price walking, for instance).
For simplicity, we assume that the coverage period is 1 year. Exposure can be included in the analysis to account for shorter coverage periods, if needed, as is done in Section 7 for the numerical illustrations.

Fairness is meaningful only for the commercial premium $\pi(\Xvec)$.
To figure out the situation, it is useful to think about the example of gender within EU, which never
enters the commercial premium since its use is prohibited but which generally
impacts on the technical premiums. The commercial premium can be regarded as
the decision made by the insurance company about the policyholder while the
technical premium corresponds to the prediction, as defined in Scantamburlo et al. (2025).
Since $\pi_0(\cdot)$
is for internal use, only, there is no need for it to be fair. On the contrary, it must be as
accurate as possible so that the insurance company is able to counteract adverse selection
and to monitor premium transfers induced within the portfolio by the application of the
commercial premiums. Precisely, any departure $\pi(\Xvec)-\pi_0(\Xvec)$ corresponds to a
transfer to (if positive) or from (if negative) another risk profile.

The paper concentrates on $\pi(\cdot)$ for which financial equilibrium and fairness
both make sense. It is nevertheless worth mentioning that the fairness notion studied in this
paper improves the performances of any candidate premium. Hence, it is also meaningful
for technical premiums to improve accuracy as measured by any Bregman loss, including
deviance.

\section{Autocalibration}

In general, autocalibration ensures that the average response in a neighborhood defined 
with the help of the autocalibrated
predictor under consideration matches the predictor value. See, e.g., Kr\"uger and Ziegel (2021).
Under autocalibration, the amounts collected from policyholders paying the same
premium match the corresponding claim totals, on average. Subsidizing transfers between policyholders are
thus avoided, helping to prevent adverse selection effects.

This concept is formally defined next, in the context of insurance.

\begin{Definition}
The commercial premium $\pi(\cdot)$ is said to be autocalibrated if 
\begin{equation}
\label{DefAutoCal}
\Esp[Y|\pi(\Xvec)=p]=p\text{ for all premium amount }p.
\end{equation}
\end{Definition}

Autocalibration is an important property in insurance pricing because 
it ensures that every group of policyholders paying the same premium is on average self-financing.
In other words, \eqref{DefAutoCal} guarantees that the premium paid by policyholders charged
the same amount of premium $p$ exactly covers the expected claims of that group. 

Notice that \eqref{DefAutoCal} can be equivalently rewritten as
$$
\Esp[Y-\pi(\Xvec)|\pi(\Xvec)=p]=\Esp[\mu(\Xvec)-\pi(\Xvec)|\pi(\Xvec)=p]=0\text{ for all premium amount }p.
$$
This shows that the net cash-flows $Y-\pi(\Xvec)$ as well as the pricing errors $\mu(\Xvec)-\pi(\Xvec)$
cancel out within groups of policyholders paying the same amount of premium, on average.
Autocalibration thus appears to be a very appealing requirement for candidate premiums.
From a conceptual point of view, \eqref{DefAutoCal} can also be written in terms of the true mean $M$
as $\Esp[M|\pi(\Xvec)]=\pi(\Xvec)$ which ensures that the commercial premium also
captures the premium that should be charged, on average.
Notice that we can equivalently condition with respect to the score in \eqref{DefAutoCal}
when $\pi(\Xvec)$ is obtained by exponentiating an additive score.

\section{Fairness notion}

\subsection{Sensitive feature}

Assume now that there is a sensitive feature $S$ recorded in the database. Its use
is allowed but requires particular attention because it exposes the insurer to
possible complaints about unfair treatment.
In this paper, we assume that $S=X_{j_0}$ for some $j_0\in\{1,2,\ldots,p\}$ to ease exposition;
see Remark \ref{SXj0} below for a discussion.
This means that the insurer actually uses $S$ for premium calculation
but wants to demonstrate that it is used in a fair and responsible way. Typical examples for $S$
include sensitive features like health or disability status, or socio-economic profile, for instance.
In the situation where individual fairness is imposed by the law,
like with gender in the EU for instance, $S$ may correspond to a feature which is correlated
to the prohibited one, as shown in the next example.

\begin{Example}
Considering the ban of gender-based discrimination in the EU, the individual fairness
imposed by the law means that the premium $\pi(\cdot)$ must fulfill the constraint
$\pi(\xvec,\text{man})=\pi(\xvec,\text{woman})=\pi(\xvec)$ for all risk profiles $\xvec$.
In words, this means that a man and a woman
with the same rating factors $\xvec$ must pay the same premium. If this condition is not
fulfilled then direct discrimination results.
This requirement is thus similar to the notions of ``fairness through unawareness'' or ``blindness'',
except that the law partly prohibits the use of proxies. According to  
Article 17 of the ``Guidelines on the application of Council Directive 2004/113/EC
to insurance, in the light of the judgment of the Court of Justice of the European Union in Case
C-236/09 (Test-Achats)'', the use of true risk factors that might be correlated with gender
remains permitted. Examples are given to illustrate the extent of the prohibition. In a nutshell,
the insurer is not allowed to include in its price list a feature without impact on claims but correlated to
gender. Hence, the use of proxy variables as substitutes for gender remains permitted, even if
it generates indirect discrimination. Here, $S$ would be a rating factor $X_{j_0}$ strongly correlated
to gender while being a true risk factor if the insurance company wishes to demonstrate that it
complies not only with the letter but also with the spirit of the anti-discrimination law.
\end{Example}

\begin{Remark}
\label{SXj0}
Even if the paper concentrates on the case $S=X_{j_0}$ for some $j_0\in\{1,2,\ldots,p\}$, the sensitive feature $S$ might also result from a combination
of several features in $\Xvec$, like disability and area, for instance, to ensure equal access to insurance for disabled
people whatever their living area. 
A similar analysis can be performed in these more complicated cases where $S=\psi(\Xvec)$ for some
given function $\psi$. Notice that $S$ may even be another prediction or premium.
For instance, $\psi(\Xvec)$ could be the prediction of policyholder's gender based on the rating
factors $\Xvec$.
\end{Remark}

\begin{Remark}\label{Rem:S-not-in-Xvec}
The case where $S$ is not contained in, or deduced from $\Xvec$ is not relevant for practice
under the fairness notion considered in this paper.
Indeed, this situation would mean either that the use of $S$ is prohibited by the law or that the
insurer refuses to let premiums depend on $S$. If we except the case where $Y$ and $S$
are independent given $\pi(\Xvec)$, any fair premium (in the sense of multicalibration
introduced below) necessarily depends on $S$
so that it conflicts with the ban imposed by the law or the commercial decision by the insurer.
\end{Remark}

\subsection{Sufficiency}

While individual fairness requires that individuals who only differ with respect to the prohibited feature
must be treated equally, group fairness is a global notion resulting from the comparison of
groups defined by the sensitive feature $S$. 
Group fairness requires that benefits or harms resulting from the adoption of a pricing structure
$\pi(\cdot)$ are distributed fairly across these groups, where the fair distribution is defined from
an appropriate group-fairness criterion.
Classical group-fairness criteria include independence, separation and sufficiency. See, e.g.,
Baumann and Loi (2023) and Charpentier (2024) for a discussion in the context of insurance.

In general, there is no consensus about the fairness criterion and different contexts
may call for different criteria.  
Moreover, fairness criteria are often incompatible so that one criterion must be selected.
In accordance with the
analysis conducted by Baumann and Loi (2023), sufficiency appears
to be an appropriate fairness condition in the context of private insurance. 
In its strict sense, sufficiency corresponds to conditional independence
of $Y$ and $S$ given $\pi(\Xvec)$. This is a rather strong requirement that rules out demographic
parity which requires independence of $\pi(\Xvec)$ and $S$.
According to Charpentier (2024, Proposition 9.3), if a candidate premium $\pi(\cdot)$ satisfies
independence and sufficiency with respect to $S$ then $Y$ and $S$ are independent.
Therefore, unless $S$ does not influence $Y$, no premium $\pi(\cdot)$ can simultaneously fulfill
demographic parity and sufficiency.

In this paper, we consider the weaker version of sufficiency retained by Baumann and Loi (2023),
defined with the help of conditional mean independence of $Y$ and $S$ given $\pi(\Xvec)$.
The conditional mean independence concept is used in econometrics to decide whether additional
features are needed in a regression model. 
It turns out that it also possesses an intuitive interpretation in the context of insurance premium.
This leads to the following definition.

\begin{Definition}
A premium $\pi(\cdot)$ is said to be fair with respect to a sensitive feature $S=X_{j_0}$
valued in $\mathcal{S}$ if
\begin{equation}
\label{DefCMI}
\Esp[Y|\pi(\Xvec)=p,S=s]=\Esp[Y|\pi(\Xvec)=p,S=s']\text{ for all $p$ and $s,s'\in\mathcal{S}$}.
\end{equation}
\end{Definition}

The requirement \eqref{DefCMI} is intuitively appealing. It just means that the expected benefits are equal
across groups created by $S$ provided the same amount of premium is paid. This argument may be used
by the insurer as a proof that the use of $S$ in premium calculation is responsible and protects
policyholders' rights because the return from the insurance operation (measured by the amount $Y$
paid by the insurer) is the same, on average, as long as the premium is identical.

If $Y\in\{0,1\}$ then \eqref{DefCMI} corresponds to the conditional independence of $Y$ and $S$
given $\pi(\Xvec)$. However, for more general responses like claim frequencies and severities,
\eqref{DefCMI} is a weaker requirement. Conditional mean independence can be tested from data. 
The reader is referred to Lundborg (2022) for a survey and to Zhang et al. (2025) for recent developments
on testing procedures.

\section{Multicalibration}

\subsection{Definition}

Given that autocalibration is a central concept in insurance pricing, rooted in financial
equilibrium, it is thus natural to combine it with conditional mean independence \eqref{DefCMI} retained as
group-fairness notion. It turns out that this corresponds to the
concept of multicalibration proposed in the machine learning literature by Hebert-Johnson et al. (2018).
This is formally defined next.

\begin{Definition}
The commercial premium $\pi(\cdot)$ is said to be multicalibrated with respect to a sensitive feature $S=X_{j_0}$
valued in $\mathcal{S}$ if
\begin{equation}
\label{DefMultiCal}
\Esp[Y|\pi(\Xvec)=p,S=s]=\Esp[Y|\pi(\Xvec)=p,S=s']=p\text{ for all $p$ and $s,s'\in\mathcal{S}$}.
\end{equation}
\end{Definition}

Compared to \eqref{DefCMI}, we see that  \eqref{DefMultiCal} not only imposes the equality of conditional
expectations, but also that they coincide with the premium amount $p$ involved in the conditions.
Notice that \eqref{DefMultiCal} can be equivalently rewritten as
$$
\Esp[Y-\pi(\Xvec)|\pi(\Xvec)=p,S=s]=0\text{ for all $p$ and $s\in\mathcal{S}$}.
$$
The net losses $Y-\pi(\Xvec)$ thus average out in any group of policyholders paying the same
amount of premium whatever the value $s$ of the sensitive feature $S$. This expresses the
compensation of claims with premiums in these groups and recalls the balance equations
at the heart of the method of marginal totals.
Also, \eqref{DefMultiCal} can be rewritten as
$$
\Esp[\mu(\Xvec)-\pi(\Xvec)|\pi(\Xvec)=p,S=s]=0\text{ for all $p$ and $s\in\mathcal{S}$}.
$$
This shows that the pricing errors $\mu(\Xvec)-\pi(\Xvec)$
cancel out on average within groups of policyholders paying the same amount of premium
and having the same value $s$ of the sensitive feature $S$.
Multicalibration thus appears to be a very appealing requirement for commercial premiums.

\subsection{Properties}

Let us first establish that multicalibration reinforces autocalibration.

\begin{Property}
\label{PropAutoMulti}
\begin{enumerate}
\item[(i)]
Any multicalibrated premium $\pi(\cdot)$ with respect to $S=X_{j_0}$ is also autocalibrated.
\item[(ii)]
If $\pi(\cdot)$ is autocalibrated then conditional mean independence \eqref{DefCMI}
implies multicalibration with respect to $S=X_{j_0}$.
\end{enumerate}
\end{Property}
\begin{proof}
For any multicalibrated premium, we have
$$
\Esp[Y|\pi(\Xvec)]=\Esp\big[\underbrace{\Esp[Y|\pi(\Xvec),S]}_{=\pi(\Xvec)}\big|\pi(\Xvec)\big]=\pi(\Xvec)
$$
so that it is also autocalibrated, as announced under item (i). Considering item (ii), notice that
\eqref{DefCMI} together with autocalibration imply
$$
\Esp[Y|\pi(\Xvec),S]=\Esp[Y|\pi(\Xvec)]=\pi(\Xvec).
$$
This ends the proof.
\end{proof}

All the properties derived for autocalibrated premiums thus also hold for multicalibrated premiums.
In particular, multicalibrated predictors never overfit.
This has been formalized for autocalibrated predictors with the help of the
convex order. Recall from Denuit et al. (2005, Section 3.4) that given two random variables $Z$ and $T$,
with finite expectations, $T$ is said to be smaller than $Z$ in the convex order (denoted as $T\lcx Z$) if 
\begin{equation}\label{eq3.A.1}
\Esp[g(T)]\leq \Esp[g(Z)]\text{ for all convex functions }g,
\end{equation}
provided the expectations exist. 
In words, $T\lcx Z$ means that $T$ is less dispersed than $Z$ while $T$ and $Z$ have the same
expected value. In particular, $\Var[T]\leq \Var[Z]$.

According to Lemma 3.1 from W\"uthrich (2023), we know
that $\pi(\Xvec)\lcx\mu(\Xvec)$ holds for any autocalibrated $\pi(\cdot)$. Therefore, any multicalibrated premium $\pi(\cdot)$ with respect to $S=X_{j_0}$ satisfies
$$
\pi(\Xvec)\lcx\mu(\Xvec)\lcx Y.
$$
This shows that multicalibration prevents overfitting, in the sense that an autocalibrated
premium cannot be more variable than $\mu(\Xvec)$ nor $Y$.

We know that $\mu(\Xvec)$ is autocalibrated as
\begin{eqnarray*}
\Esp[Y|\mu(\Xvec)]&=&\Esp\big[\Esp[Y|\mu(\Xvec),\Xvec]\big|\mu(\Xvec)\big]\\
&=&\Esp\big[\Esp[Y|\Xvec]\big|\mu(\Xvec)\big]=\mu(\Xvec).
\end{eqnarray*}
The next result shows that $\mu(\Xvec)$ is also
multicalibrated with respect to $S=X_j$ for every $j\in\{1,\ldots,p\}$.

\begin{Property}
\label{PropMultiMu}
The best-estimate price $\mu(\Xvec)$ is multicalibrated with respect to any $S=X_j$, $j=1,2,\ldots,p$.
\end{Property}
\begin{proof}
It suffices to write
\begin{eqnarray*}
\Esp[Y|\mu(\Xvec),S]&=&\Esp\big[\Esp[Y|\mu(\Xvec),S,\Xvec]\big|\mu(\Xvec),S\big]\\
&=&\Esp\big[\Esp[Y|\Xvec]\big|\mu(\Xvec),S\big]\\
&=&\mu(\Xvec),
\end{eqnarray*}
which ends the proof.
\end{proof}

Since pure premiums correspond to conditional expectations, they can thus be
consistently estimated only if the expected loss is minimum for the mean response.
The class of Bregman loss functions is known to be strictly consistent for the
(conditional) mean functional so that they are the only meaningful loss functions
in the context of insurance pricing.
For a convex function $\ell(\cdot)$, recall that the Bregman loss function $L(\cdot,\cdot)$
is defined as
\begin{equation}
\label{BregmanLoss}
L(y, m) = \ell(y)-\ell(m) - \ell '(m)(y-m)
\end{equation}
where $\ell'$ denotes the subgradient of the convex function $\ell$.
A premium $\pi_2(\cdot)$ outperforms $\pi_1(\cdot)$ for the Bregman loss function $L$ in \eqref{BregmanLoss}
if $\Esp[L(Y,\widehat{\pi}_2(\Xvec))]\leq \Esp[L(Y,\widehat{\pi}_1(\Xvec))]$.
Since multicalibrated predictors are also autocalibrated, we know from Theorem 3.1
in Kr\"uger and Ziegel (2021) that given two multicalibrated premiums $\pi_1(\cdot)$ and $\pi_2(\cdot)$,
the inequality $\Esp[L(Y,\pi_2(\Xvec))]\leq \Esp[L(Y,\pi_1(\Xvec))]$
holds true for every Bregman loss function $L$ in \eqref{BregmanLoss}
if, and only if, $\pi_1(\Xvec)\lcx \pi_2(\Xvec)$.
This shows that the convex order is the appropriate tool to compare the relative performances
of multicalibrated premiums.

\section{Multibalance correction}

Balance correction has been proposed as a practical way to restore autocalibration
by Denuit et al. (2021). Precisely,
the balance-corrected version $\pi_{\mathrm{bc}}(\cdot)$ of the candidate premiums $\pi(\cdot)$ 
is defined as
\begin{equation}
\label{DefBC}
\pi_{\mathrm{bc}}(\Xvec)=\Esp[Y|\pi(\Xvec)].
\end{equation}
The new premium $\pi_{\mathrm{bc}}(\cdot)$ obtained from \eqref{DefBC} is always autocalibrated
as shown in Proposition 3.3 in W\"uthrich (2023) without the monotonicity condition assumed
in Denuit et al. (2021).  The intuition behind  \eqref{DefBC} is that $\pi(\cdot)$ is informative
but not necessarily on the right scale, so that financial equilibrium defining autocalibration is violated.
The conditional expectation \eqref{DefBC} corrects the candidate premium by retaining the order
induced by $\pi(\cdot)$ but averaging insurance losses corresponding to the same value of
$\pi(\Xvec)$.

Inspired from balance correction \eqref{DefBC} restoring autocalibration, let
us introduce multibalance correction restoring multicalibration.

\begin{Definition}
The multibalance-corrected version $\pi_{\mathrm{mbc}}(\cdot)$ of the premium $\pi(\cdot)$ 
with respect to $S$ is defined as
\begin{equation}
\label{DefMBC}
\pi_{\mathrm{mbc}}(\Xvec)=\Esp[Y|\pi(\Xvec),S].
\end{equation}
\end{Definition}

The next result shows that the premium $\pi_{\mathrm{mbc}}(\cdot)$ obtained by multibalance correction \eqref{DefMBC} is always multicalibrated with respect to $S$.

\begin{Proposition}
The multibalance-corrected version $\pi_{\mathrm{mbc}}(\cdot)$ of the premium $\pi(\cdot)$  defined in \eqref{DefMBC}
is multicalibrated with respect to $S=X_{j_0}$.
\end{Proposition}
\begin{proof}
The announced result follows from
\begin{eqnarray*}
\Esp[Y|\pi_{\mathrm{mbc}}(\Xvec),S]&=&\Esp\big[Y\big|\Esp[Y|\pi(\Xvec),S],S\big]\\
&=&\Esp\Big[\Esp\big[Y\big|\Esp[Y|\pi(\Xvec),S],S,\pi(\Xvec)\big]\Big|\Esp[Y|\pi(\Xvec),S],S\Big]\\
&=&\Esp\Big[\Esp\big[Y\big|S,\pi(\Xvec)\big]\Big|\Esp[Y|\pi(\Xvec),S],S\Big]\\
&=&\Esp\big[\pi_{\mathrm{mbc}}(\Xvec)\big|\pi_{\mathrm{mbc}}(\Xvec),S\big]\\
&=&\pi_{\mathrm{mbc}}(\Xvec),
\end{eqnarray*}
which ends the proof.
\end{proof}

It turns out that multibalance correction improves the performance of the premium
with respect to any Bregman loss function. This is precisely stated next.

\begin{Proposition}
\label{PropBregLoss}
For any premium $\pi$, we have
$$
\Esp\big[L\big(Y,\pi(\Xvec)\big) \big] \geq \Esp\big[L\big(Y,\pi_{\mathrm{bc}}(\Xvec)\big) \big] 
\geq \Esp\big[L\big(Y,\pi_{\mathrm{mbc}}(\Xvec)\big) \big] 
 \geq \Esp\big[L\big(Y,\mu(\Xvec)\big) \big]
$$
for all loss functions \eqref{BregmanLoss}.
\end{Proposition}
\begin{proof}
Since $\Esp[Y|\pi(\Xvec)]\lcx \Esp[Y|\pi(\Xvec),S]$, we have $\pi_{\mathrm{bc}}(\Xvec)\lcx\pi_{\mathrm{mbc}}(\Xvec)\lcx \mu(\Xvec)$. The announced result then follows from Proposition 4.5 in Denuit and Trufin (2024)
and the equivalence of Bregman dominance and convex order for autocalibrated predictors.
\end{proof}

\begin{Remark}
The multibalance-corrected premium obtained from a candidate premium $\pi(\cdot)$
does not necessarily coincide with the multibalance-corrected premium obtained from its
balance-corrected version $\pi_{\mathrm{bc}}(\cdot)$.
Indeed, 
\begin{eqnarray*}
\Esp[Y|\pi_{\mathrm{bc}}(\Xvec),S]
&=&
\Esp\big[\Esp[Y|\pi(\Xvec),\pi_{\mathrm{bc}}(\Xvec),S]\big|\pi_{\mathrm{bc}}(\Xvec),S\big] \\
&=&
\Esp\big[\Esp[Y|\pi(\Xvec),S]\big|\pi_{\mathrm{bc}}(\Xvec),S\big]\\
&=&
\Esp\big[\pi_{\mathrm{mbc}}(\Xvec)\big|\pi_{\mathrm{bc}}(\Xvec),S\big].
\end{eqnarray*}
Hence, the multibalance-corrected premium based on $\pi_{\mathrm{bc}}(\cdot)$ is obtained by
averaging the multibalance-corrected premium based on $\pi(\cdot)$ over the level sets of
$\pi_{\mathrm{bc}}(\cdot)$, for each value of $S$.
As a consequence, the two multibalance-corrected premiums may differ, since this
averaging can remove information relevant for predicting $Y$ once $S$ is given.
They coincide if, and only if, conditioning with respect to $\pi_{\mathrm{bc}}(\Xvec)$ retains all the
predictive information carried by $\pi(\Xvec)$ for each value of $S$
(that is, if and only if $\pi(\Xvec)$ is measurable with respect to
$\sigma(\pi_{\mathrm{bc}}(\Xvec),S)$).
A simple sufficient condition for this to hold is that the mapping
$p\mapsto\Esp[Y|\pi(\Xvec)=p]$ be one-to-one on the range of $\pi(\Xvec)$, so that
$\pi(\Xvec)$ can be recovered from $\pi_{\mathrm{bc}}(\Xvec)$.
In particular, this condition is satisfied when the mapping
$p\mapsto\Esp[Y|\pi(\Xvec)=p]$ is strictly increasing, which means that $\pi(\Xvec)$
induces the same ordering of risks as $\mu(\Xvec)=\Esp[Y|\Xvec]$.
\end{Remark}

To conclude this section, we consider the case discussed in Remark~\ref{Rem:S-not-in-Xvec},
where the sensitive feature $S$ is excluded from the set of admissible
rating factors, either because of regulation or of a commercial decision made by the insurer.
The next result shows how multibalance correction provides an illustration
of the structural limitation highlighted in Remark~\ref{Rem:S-not-in-Xvec}.

\begin{Proposition}
\label{Prop:MBC-general}
Assume that $S$ is excluded from the set of admissible rating factors $\Xvec$.
Then, for any admissible premium $\pi(\cdot)$, the following statements hold.
\begin{enumerate}
\item[(i)]
The corresponding premium $\pi_{\mathrm{mbc}}(\cdot)$ is multicalibrated with respect to $S$, in the sense that
\[
\Esp[Y|\pi_{\mathrm{mbc}}(\Xvec,S),S]=\pi_{\mathrm{mbc}}(\Xvec,S).
\]

\item[(ii)]
In general, $\pi_{\mathrm{mbc}}(\cdot)$ depends on $S$ and therefore does not define an admissible
premium.

\item[(iii)]
The multibalance-corrected premium $\pi_{\mathrm{mbc}}(\cdot)$ does not depend on $S$ if, and only if,
\[
\Esp[Y|\pi(\Xvec),S]=\Esp[Y|\pi(\Xvec)],
\]
that is, if, and only if, $Y$ and $S$ are conditionally mean-independent, given the premium $\pi(\Xvec)$.

\item[(iv)]
The multibalance correction can be written as
\[
\pi_{\mathrm{mbc}}(\Xvec,S)=\Esp\big[\Esp[Y|\Xvec,S]\big|\pi(\Xvec),S\big].
\]
In particular, $\pi_{\mathrm{mbc}}(\cdot)$ coincides with $\Esp[Y|\Xvec,S]$ if, and only if, $\Esp[Y|\Xvec,S]$ is measurable with respect to $\sigma(\pi(\Xvec),S)$.
\end{enumerate}
\end{Proposition}

\begin{proof}
By definition of the multibalance correction, $\pi_{\mathrm{mbc}}(\Xvec,S)=\Esp[Y|\pi(\Xvec),S]$.
Therefore,
\begin{eqnarray*}
\Esp[Y|\pi_{\mathrm{mbc}}(\Xvec,S),S]
&=&
\Esp\big[\Esp[Y|\pi(\Xvec),\pi_{\mathrm{mbc}}(\Xvec,S),S]\big|\pi_{\mathrm{mbc}}(\Xvec,S),S\big] \\
&=&
\Esp\big[\Esp[Y|\pi(\Xvec),S]\big|\pi_{\mathrm{mbc}}(\Xvec,S),S\big]\\
&=&\Esp[\pi_{\mathrm{mbc}}(\Xvec,S)|\pi_{\mathrm{mbc}}(\Xvec,S),S]\\
&=&
\pi_{\mathrm{mbc}}(\Xvec,S).
\end{eqnarray*}
This proves the statement under item (i).

Considering items (ii)-(iii), notice that $\pi_{\mathrm{mbc}}(\Xvec,S)=\Esp[Y|\pi(\Xvec),S]$
generally depends on $S$.
Hence, $\pi_{\mathrm{mbc}}(\cdot)$ is independent of $S$ if, and only if,
\[
\Esp[Y|\pi(\Xvec),S]=\Esp[Y|\pi(\Xvec)],
\]
which proves the statements under items (ii)-(iii).

Finally, 
\begin{eqnarray*}
\Esp[Y|\pi(\Xvec),S]
&=&\Esp\big[\Esp[Y|\Xvec,\pi(\Xvec),S]\big|\pi(\Xvec),S\big] \\
&=&\Esp\big[\Esp[Y|\Xvec,S]\big|\pi(\Xvec),S\big]
\end{eqnarray*}
proves the statement under item (iv) and ends the proof.
\end{proof}

Proposition~\ref{Prop:MBC-general} provides an
illustration of Remark~\ref{Rem:S-not-in-Xvec}.
When $S$ is excluded from the set of admissible rating factors $\Xvec$, multicalibration
can only be achieved in the exceptional case where $Y$ and $S$ are conditionally mean-independent,
given the premium.
Otherwise, enforcing multicalibration necessarily leads to premiums that depend on $S$
and are thus inadmissible.

\section{Practical implementation of multicalibration}

This section proposes practical ways to make a candidate premium $\pi(\cdot)$ multicalibrated
with respect to a sensitive feature $S$ included in $\Xvec$. The approaches differ according to the format of $S$.

\subsection{Categorical sensitive feature}
\label{subsec:catS}

Assume that the sensitive feature $S$ is categorical, with levels $\ell\in\mathcal{S}$.
In principle, multicalibration could be implemented by making
$\pi(\cdot)$ autocalibrated separately in every group $S=\ell$,
by replacing $\pi(\cdot)$ with its balance-corrected version $\pi_{\text{bc}}(\cdot)$
for every level of $S$. However, this simple
approach is problematic if exposures are limited for some groups, so that non-parametric
smoothing does not work.

This section presents two empirical approaches to multicalibration.
The first one implements multicalibration through an explicit multibalance correction, while the second one enforces multicalibration directly through an iterative bias-correction procedure.
Throughout this section, we consider a dataset of observations $(N_i,\Xvec_i,S_i,w_i)$, $i=1,\ldots,n$, where $N_i$ denotes the claim count, $S_i=X_{i,j_0}\in\mathcal S$ the categorical sensitive feature and $w_i\ge0$ known weights (typically exposures). In practice, we use the observed claim frequency $Y_i=\frac{N_i}{w_i}$ as response variable, obtained by dividing the claim count by the exposure.

\subsubsection{Direct empirical multibalance correction via groupwise regression}
\label{subsec:iso_mbc}

Multibalance correction amounts to restoring balance separately within each group defined
by the sensitive feature $S$.
That is, instead of enforcing autocalibration globally, the balance condition is imposed
conditionally on $S=\ell$, for every $\ell\in\mathcal S$.
This principle can be implemented in several ways, depending on how the conditional mean
function $\Esp[Y | \pi(\Xvec),S=\ell]$ is estimated within each group.
Examples include local linear regression as in Ciatto et al.\ (2023) or isotonic regression
as in W\"uthrich (2023).

To fix ideas, we describe here an implementation based on isotonic regression.
For each level $\ell\in\mathcal S$, consider the conditional mean function
\[
m_\ell(p)=\Esp\big[Y | \pi(\Xvec)=p, S=\ell\big],\qquad p>0.
\]
An empirical implementation of multibalance correction is obtained by estimating $m_\ell(\cdot)$
nonparametrically within each group $S=\ell$, under a monotonicity constraint in $p$.
Specifically, we define $\widehat m_\ell$ as the solution of the weighted isotonic regression problem
\begin{equation}
\label{eq:isotonic_groupwise}
\widehat m_\ell \in \argmin_{m\ \mathrm{non-decreasing}}
\sum_{i:\,S_i=\ell} w_i\bigl(Y_i-m(\pi_i)\bigr)^2,
\qquad \ell\in\mathcal S,
\end{equation}
restricting the search to non-decreasing conditional mean function, that is, such that $p\le p'$ implies $m(p)\le m(p')$ on the range of
$\{\pi_i:S_i=\ell\}$.
The resulting empirical multibalance-corrected premium is then obtained by the plug-in rule
\begin{equation}
\label{eq:pi_mbc_hat_isotonic}
\widehat\pi_{\mathrm{mbc}}(\Xvec_i)
=
\widehat m_{S_i}\!\big(\pi(\Xvec_i)\big),
\qquad i=1,\ldots,n,
\end{equation}
where $\widehat{m}_{\ell}(\cdot)$ is the solution to \eqref{eq:isotonic_groupwise}.

This construction is the direct analogue of balance correction for autocalibration implemented
via isotonic regression by W\"uthrich (2023), except that the recalibration is performed separately within each
level of the categorical sensitive feature $S$.
By construction, the premium $\widehat\pi_{\mathrm{mbc}}(\cdot)$ in
\eqref{eq:pi_mbc_hat_isotonic} is an empirical implementation of the multibalance-corrected
premium
\[
\pi_{\mathrm{mbc}}(\Xvec)=\Esp[Y | \pi(\Xvec),S],
\]
and therefore achieves multicalibration by enforcing balance conditionally on $S$.

When some levels $\ell\in\mathcal S$ carry limited exposure, the fully nonparametric estimation in
\eqref{eq:isotonic_groupwise} may be unstable, yielding step functions with high variance and
poor out-of-sample behavior.
The next section therefore introduces a regularized procedure that relaxes the fully
non-parametric group-wise approach.
Rather than estimating conditional mean functions independently within each level of $S$,
the procedure enforces balance on a finite collection of $(\pi,S)$-cells and stabilizes
local corrections by borrowing strength across groups through exposure-driven shrinkage.

\subsubsection{Regularized empirical enforcement of multicalibration via iterative bias correction}
\label{subsec:iter_multi}

We now introduce a regularized empirical procedure that enforces multicalibration directly,
without explicitly applying the multibalance correction.
Starting from a baseline premium, the procedure modifies the premium iteratively so that,
at the chosen level of discretization, average residuals vanish within each group defined by
the current premium level and the sensitive feature.
Exposure-driven shrinkage is used to stabilize local corrections when data are sparse.

\paragraph{Binning and residual biases.}
Fix an integer $K\ge 1$.
Starting from an initial premium $\pi^{(0)}(\cdot)=\pi(\cdot)$, consider at iteration $j\ge 0$
a partition of the range of $\pi^{(j)}(\Xvec)$ into $K$ disjoint intervals
$\{B^{(j)}_1,\dots,B^{(j)}_K\}$ (for instance empirical quantile bins).
For each observation $i=1,\ldots,n$, define the residual
\[
R_i^{(j)}=Y_i-\pi^{(j)}(\Xvec_i).
\]
For each bin $k\in\{1,\ldots,K\}$ and each level $\ell\in\mathcal S$, define the index set
\[
I^{(j)}_{k\ell}
=
\big\{ i:\ \pi^{(j)}(\Xvec_i)\in B^{(j)}_k,\ S_i=\ell \big\},
\]
together with the corresponding exposure
\[
w^{(j)}_{k\ell}=\sum_{i\in I^{(j)}_{k\ell}} w_i,
\qquad
w^{(j)}_{k\bullet}=\sum_{\ell\in\mathcal S} w^{(j)}_{k\ell}.
\]
Whenever $w^{(j)}_{k\ell}>0$, define the cell-wise and pooled bin-wise mean residuals as
\[
\widehat b^{(j)}_{k\ell}
=
\frac{1}{w^{(j)}_{k\ell}}\sum_{i\in I^{(j)}_{k\ell}} w_i R_i^{(j)},
\qquad
\widehat b^{(j)}_{k}
=
\frac{1}{w^{(j)}_{k\bullet}}\sum_{\ell\in\mathcal S}\ \sum_{i\in I^{(j)}_{k\ell}} w_i R_i^{(j)}.
\]
The quantity $\widehat b^{(j)}_{k\ell}$ is the exposure-weighted sample analogue of $\Esp\!\left[Y-\pi^{(j)}(\Xvec)\,\big|\,\pi^{(j)}(\Xvec)\in B^{(j)}_k,\ S=\ell\right]$.

\paragraph{Exposure-driven shrinkage.}
When $w^{(j)}_{k\ell}$ is small, the estimate $\widehat b^{(j)}_{k\ell}$ is unstable.
To stabilize the correction, we shrink $\widehat b^{(j)}_{k\ell}$ toward the pooled bin-level quantity
$\widehat b^{(j)}_{k}$ according to
\begin{equation}
\label{eq:shrink_bias}
\widetilde b^{(j)}_{k\ell}
=
z^{(j)}_{k\ell}\widehat b^{(j)}_{k\ell}
+
\bigl(1-z^{(j)}_{k\ell}\bigr)\widehat b^{(j)}_k,
\qquad
z^{(j)}_{k\ell}=\frac{w^{(j)}_{k\ell}}{w^{(j)}_{k\ell}+c},
\end{equation}
where $c>0$ is a tuning parameter controlling the amount of pooling across groups.

\paragraph{Update step.}
Let $\eta\in(0,1]$ be a step size.
For each observation $i$ such that $\pi^{(j)}(\Xvec_i)\in B^{(j)}_k$ and $S_i=\ell$, update
\begin{equation}
\label{eq:update_iter_mbc}
\pi^{(j+1)}(\Xvec_i)
=
\pi^{(j)}(\Xvec_i)
+
\eta\,\widetilde b^{(j)}_{k\ell}.
\end{equation}
This update adjusts the premium by an exposure-regularized estimate of the residual bias
in the corresponding $(\pi,S)$-cell.

\paragraph{Stopping criterion.}
Iterations are stopped when the remaining correction becomes negligible.
A convenient criterion is to retain the solution such that
\[
\max_{k,\ell:\,w^{(j)}_{k\ell}>0}
\frac{\big|\eta\,\overline{\widetilde b}{}^{(j)}_{k\ell}\big|}
{\overline{\pi}{}^{(j)}_{k\ell}}
\le \delta,
\]
where $\overline{\pi}{}^{(j)}_{k\ell}$ and $\overline{\widetilde b}{}^{(j)}_{k\ell}$
denote exposure-weighted averages of $\pi^{(j)}(\Xvec_i)$ and $\widetilde b^{(j)}_{k\ell}$
over $i\in I^{(j)}_{k\ell}$, and $\delta>0$ is a tolerance parameter selected by the
actuary.

\paragraph{Interpretation.}
At convergence, the procedure yields a premium $\pi^{(J)}(\cdot)$ such that, for each bin $B^{(J)}_k$ and each level $\ell\in\mathcal S$, the exposure-weighted mean residual within the corresponding cell is close to zero. Equivalently, at the resolution induced by the binning scheme, no systematic average deviation between $Y$ and $\pi^{(J)}(\Xvec)$ remains within groups defined by $(\pi^{(J)}(\Xvec),S)$.

This property is the discrete counterpart of the multicalibration requirement
$  \Esp[Y | \pi^{(J)}(\Xvec)=p,\ S=\ell]=p$, $p>0$, $\ell\in\mathcal S$,
with exact conditioning replaced by conditioning on bin membership. Because the bins are updated along the iterations and are defined using the current premium $\pi^{(j)}(\cdot)$, the resulting premium $\pi^{(J)}(\cdot)$ should be understood as enforcing multicalibration through a fixed-point property, rather than as the result of a single application of the multibalance correction.

In contrast, if the bins are fixed using the baseline premium $\pi(\cdot)$ and a single update is performed, the procedure reduces to a discretized empirical approximation of the multibalance-corrected premium
$ \pi_{\mathrm{mbc}}(\Xvec)=\Esp[Y | \pi(\Xvec),S]$, at the chosen binning resolution.

\paragraph{On the role of credibility weights.}

In the classical B\"uhlmann--Straub framework, credibility arises from an explicit stochastic data-generating structure. A latent parameter $\Theta$ governs repeated observations at individual risk level, leading to the well-known variance decomposition into process variance and structural variance.

In the present multicalibration setting, no such structural latent parameter is postulated. The procedure is entirely cross-sectional and algorithmic. For a given premium bin $k$ and group $\ell$, the quantity
\begin{align*}
b_{k\ell} = \Esp\!\left[Y-\pi(X)| \pi(X)\in B_k,\; S=\ell\right]
\end{align*}
may be viewed as a latent systematic bias parameter attached to the cell $(k,\ell)$, but this parameter is induced by the current premium rather than by an underlying stochastic heterogeneity. The observations within each cell simply play the role of repeated measurements used to estimate this conditional mean residual.

The analogy with credibility theory is therefore algebraic rather than structural. As in the B\"uhlmann--Straub model, the empirical estimator $\widehat b_{k\ell}$ is noisy when exposure $w_{k\ell}$ is small. A variance decomposition of the form
\begin{align*}
\Var[\widehat b_{k\ell}] = \frac{\Esp[\sigma^2_{k\ell}]}{w_{k\ell}} + \Var[b_{k\ell}]
\end{align*}
motivates shrinking the cell-specific estimator toward the pooled bin-level quantity $\widehat b_k$, leading to the weight
\begin{align*}
z_{k\ell} = \frac{w_{k\ell}}{w_{k\ell}+c}.
\end{align*}

The resulting correction
\begin{align*}
\widetilde b_{k\ell} = z_{k\ell}\widehat b_{k\ell} + (1-z_{k\ell})\widehat b_k
\end{align*}
should therefore be interpreted primarily as an exposure-driven regularization device. It stabilizes local bias estimates in low-exposure cells while allowing more granular adjustments where sufficient data are available.

In contrast with classical credibility theory, however, we do not assume an explicit hierarchical model with a stochastic latent parameter $\Theta$. The credibility weight here provides a principled and actuarially interpretable shrinkage mechanism, but the multicalibration procedure itself remains a post-processing algorithm rather than the consequence of a fully specified
probabilistic credibility model.

\subsubsection{Numerical illustrations}
\label{subsec:illustrdisc}

We apply the proposed multicalibration procedure to the data set used by
Noll et al. (2018). This data set contains a French motor third-party liability (MTPL) insurance portfolio comprising of $677~991$ entries, each corresponding to a policy. Twelve variables are observed for each of these policies, listed in Table \ref{tab:variables}. We refer to Noll et al. (2018) for a detailed description of the complete data set.  

\begin{table}[h!]
\begin{center}
\caption{Description of the variables in the \texttt{freMTPL2freq} dataset.}
\label{tab:variables}
\begin{tabular}{@{}lp{0.85\textwidth}@{}}
\toprule
Variable & Description \\ \midrule
\texttt{IDpol} & Policy number. \\
\texttt{ClaimNb} & Number of claims. \\
\texttt{Exposure} & Total exposure in yearly units. \\
\texttt{Area} & Area code (A,B,C,D,E or F).\\
\texttt{Region} & Region of France where the policy is registered, categorical variable with 22 levels.\\
\texttt{Density} & Population density (number of inhabitants per square km) of the city where the policyholder lives. \\
\texttt{BonusMalus} & Bonus-malus level, ranging between 50 and 230 (reference level: 100).\\
\texttt{DrivAge} & Age of the driver of the vehicle, in years.\\
\texttt{VehAge} & Age of the vehicle, in years.\\
\texttt{VehGas} & Type of fuel used for the vehicle (diesel or regular gas).\\
\texttt{VehPower} & Power of the car, categorical variable ranging from 4 to 15. \\
\texttt{VehBrand} & Brand of the car, categorical variable with 11 levels. \\
\bottomrule
\end{tabular}
\end{center}
\end{table}

The outcome variable $Y$ is the observed claim frequency, defined as $Y=\frac{N}{w}$, where $N=\texttt{ClaimNb}$ and $w=\texttt{Exposure}$ measured in policy-years. For the multicalibration analysis, we consider a discrete sensitive attribute
\begin{align*}
    S^{\mathrm d}=\texttt{bin\_VehAge},
\end{align*}
obtained by grouping \texttt{VehAge} into three bins of approximately equal size: $(0,3]$, $(3,9]$ and $>9$. The rating variables are
\begin{align*}
    \Xvec = \{\texttt{Area}, \texttt{Region}, \texttt{Density}, \texttt{BonusMalus}, \texttt{DrivAge}, \texttt{VehAge}, \texttt{VehGas}, \texttt{VehBrand} \},
\end{align*}
while $\texttt{bin\_VehAge}$ is used solely for calibration.

The data are randomly split into training (60\%), validation (20\%), and test (20\%) sets. Baseline premiums are obtained from a Poisson GAM with log-link and exposure offset. Precisely, we assume that
\begin{align*}
    N_i | \Xvec_i \sim \mathrm{Poisson}\big(w_i \lambda(\Xvec_i)\big),
\end{align*}
with
\begin{align*}
    \ln \lambda(\Xvec_i) = \beta_0 + \sum_m f_m(X_{i,m}).
\end{align*}
The corresponding frequency premium is then given by
\begin{align*}
    \pi(\Xvec_i)=\lambda(\Xvec_i)=\Esp[Y_i| \Xvec_i].
\end{align*}
Although $\pi(\cdot)$ is calibrated on average, it displays systematic residual bias across premium levels and across vehicle-age groups.

Starting from $\pi(\cdot)$, we consider five premium constructions. The first is the unawareness premium itself. The second corresponds to autocalibration obtained as a limiting case of the iterative multicalibration procedure when the credibility parameter $c$ is taken to be very large. In that regime,
\begin{align*}
    z_{k\ell} = \frac{w_{k\ell}}{w_{k\ell}+c} \approx 0,
\end{align*}
so that only the bin-level bias enters the update, enforcing autocalibration.
As a non-iterative benchmark targeting the same objective, we also consider classical balance correction
$ \pi_{\mathrm{bc}}(\Xvec)=\Esp[Y | \pi_0(\Xvec)]$ based on weighted isotonic regression,
Full multicalibration is obtained for finite $c$, combining bin-level and cell-level biases through exposure-driven shrinkage so as to enforce $  \Esp[Y-\pi(\Xvec) | \text{bin}, S] \approx 0$.
Finally, multibalance correction is defined by
$  \pi_{\mathrm{mbc}}(\Xvec)=\Esp[Y | \pi(\Xvec), S]$,
estimated via weighted isotonic regression within each vehicle-age group.
Throughout, we use $K=10$ premium bins, step size $\eta=0.2$, and stopping threshold $\delta=0.01$. Residual bias is computed as exposure-weighted averages within each premium bin and vehicle-age group.

Figure~\ref{fig:MCvsMBC} illustrates the residual structure. The baseline premiums exhibit a clear monotone distortion across premium bins, with pronounced under-pricing in the highest bins. After autocalibration (iterative or isotonic), this global pattern disappears and the overall residual curve becomes essentially flat. However, significant differences remain across vehicle-age groups within the same premium bin, indicating persistent conditional bias with respect to $S$, as shown in the plot on the second row.

Applying full multicalibration or multibalance correction removes these group-specific distortions. In both cases, the colored curves align closely around zero across all bins. Multicalibration achieves this through iterative bias updates with exposure-driven shrinkage, producing smooth residual profiles. Multibalance correction enforces conditional calibration directly via group-wise isotonic regression. Minor local oscillations reflect the step-wise nature of the estimator and the absence of cross-group regularization. Overall, both methods successfully restore conditional balance while preserving premium-level calibration, with the iterative approach yielding slightly more regular patterns.

\begin{figure}[h!]
    \centering
    \includegraphics[width=\textwidth]{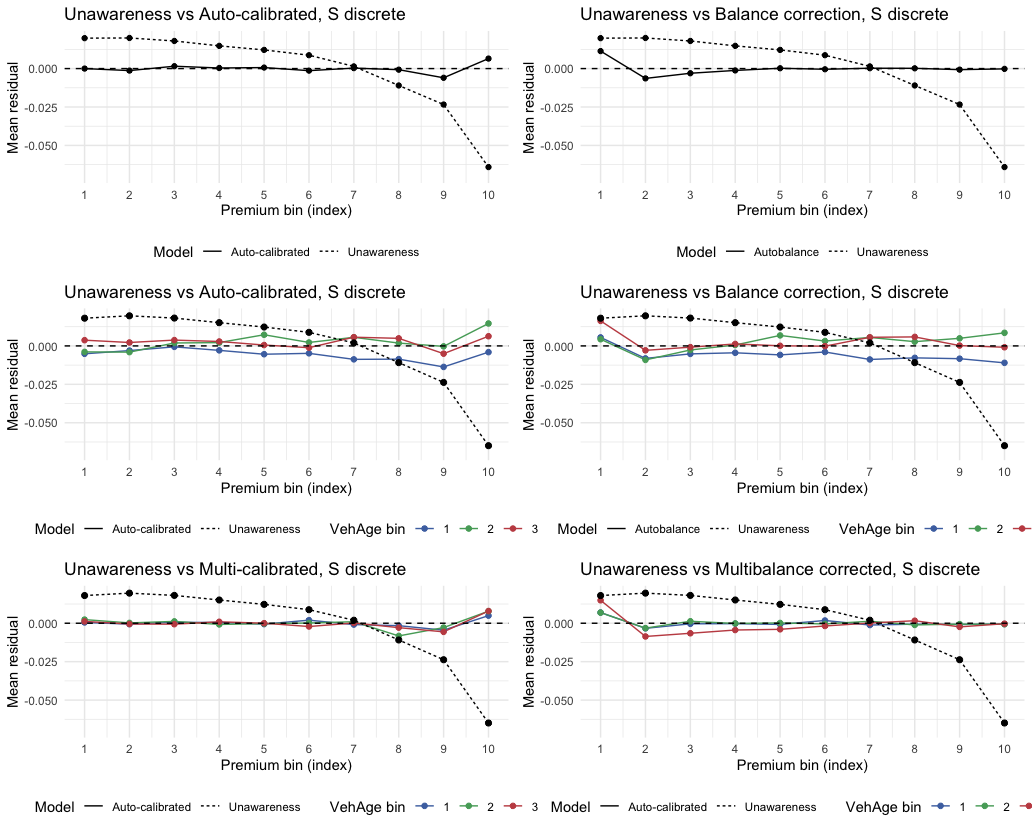}
    \caption{Mean residual pricing bias across premium bins for discrete vehicle-age groups. Top row: baseline vs.\ autocalibration. Middle row: autocalibration via iterative algorithm (left) and balance correction via isotonic regression (right). Bottom row: multicalibration (iterative) and multibalance correction (group-wise isotonic regression).}
    \label{fig:MCvsMBC}
\end{figure}

Table~\ref{tab:perf_test_disc} reports out-of-sample performance on the test set using Poisson deviance and Gini coefficient. The baseline model exhibits by far the highest deviance, confirming that miscalibration in the premium score materially affects predictive accuracy. Although the baseline GAM is estimated by maximum likelihood in the covariate space, it does not guarantee that the induced premium satisfies autocalibration. Systematic distortions across premium levels therefore translate into likelihood losses when aggregated over the portfolio.

Autocalibration yields a substantial reduction in deviance relative to the baseline, showing that correcting premium-level distortions already improves adequacy. Introducing conditional calibration with respect to vehicle age further reduces deviance: the multicalibrated premium achieves the lowest deviance overall, improving upon the autocalibrated specification. In comparison, the non-iterative balance and multibalance corrections achieve slightly higher deviance values. Overall, the iterative procedures (autocalibration and multicalibration) outperform their isotonic counterparts in terms of deviance.

In terms of discriminatory power, all calibrated models substantially improve upon the baseline. The Gini coefficients of the calibrated models are relatively close, with autocalibration and balance correction attaining the highest values. 

\begin{table}[htbp]
\centering
\caption{Out-of-sample performance comparison on the test set.}
\label{tab:perf_test_disc}
\begin{tabular}{lcc}
\hline
Model & Deviance & Gini Coefficient \\
\hline
Baseline & 36\,310.68 & 0.5783  \\
Autocalibrated & 33\,882.79  & \textbf{0.7237}\\
Balance correction  & 33\,912.20 & 0.7221 \\
Multicalibrated & \textbf{33\,779.73} & 0.6876  \\
Multibalance correction & 33\,808.72 & 0.6938  \\
\hline
\end{tabular}
\end{table}

\subsection{Continuous sensitive feature}

We now consider the case where the sensitive feature $S$ is continuous. In this setting, multicalibration enforces $\Esp[Y| \pi(\Xvec)=p,\;S=s]=p$, 
for $(p,s)$ in the support of $(\pi(\Xvec),S)$.
In analogy with the categorical case, we present two empirical approaches. The first one implements multicalibration through a direct estimation of the conditional mean surface in the joint $(\pi,S)$ space. The second one enforces multicalibration through a regularized iterative bias-correction procedure.
We work with the same data structure as before. 

\subsubsection{Direct empirical multibalance correction via bivariate local GLM}
\label{subsec:biv_mbc}

When the sensitive feature $S$ is continuous, we estimate the conditional mean surface
\begin{align*}
    m(p,s) = \Esp\!\left[Y | \pi(\Xvec)=p, S=s\right],
\end{align*}
through a local GLM in the joint $(\pi,S)$ space.

To estimate $m(p,s)$ at a given $(p,s)$, we assign to each observation $(N_i,\pi_i,S_i,w_i)$ a weight
\begin{align*}
    \nu_i(p,s) = \nu\!\left(\frac{\pi_i-p}{h_\pi},\frac{S_i-s}{h_S}\right),
\end{align*}
where $\nu(\cdot,\cdot)$ is a symmetric weight function supported on a bounded set and $h_\pi$
and $h_S>0$ are bandwidth parameters controlling locality in each direction.
At $(p,s)$, we fit an intercept-only GLM adapted to the nature of the response, using exposure weights $w_i$. Under the canonical link function $g$, we obtain the local likelihood equation
\begin{align*}
    \sum_{i=1}^n \nu_i(p,s)\,w_i\,Y_i = \sum_{i=1}^n \nu_i(p,s)\,w_i\,\mu(p,s),
\end{align*}
whose solution yields
\begin{align*}
    \widehat m(p,s) = \frac{\sum_{i=1}^n \nu_i(p,s)\,w_i\,Y_i}{\sum_{i=1}^n \nu_i(p,s)\,w_i}.
\end{align*}

The response variable is again the observed claim frequency $Y=N/w$, and estimation is performed using an exposure-weighted local GLM. A natural plug-in estimator replaces $\pi(\Xvec_i)$ with $\widehat m\!\big(\pi(\Xvec_i),S_i\big)$. 
While this estimator targets conditional balance with respect to $(\pi,S)$, it does not generally preserve marginal autocalibration with respect to $\pi$.
To restore marginal consistency, we therefore introduce a centering step. First, we estimate the marginal conditional mean
\begin{align*}
    \widehat m_0(p) = \Esp\!\left[Y | \pi(\Xvec)=p\right]
\end{align*}
using a one-dimensional local GLM in the premium score. We then define
\begin{align*}
    \delta(p,s) = \widehat m(p,s) - \widehat m_0(p),
\end{align*}
and estimate its conditional expectation given $p$ via a one-dimensional local regression. Denoting this estimate by $\widehat{\Esp}[\delta(p,S)| \pi=p]$, we construct the centered multibalance correction
\begin{align*}
    \widehat\pi_{\mathrm{mbc}}(\Xvec_i)
    =
    \widehat m_0\!\big(\pi(\Xvec_i)\big)
    +
    \delta\!\big(\pi(\Xvec_i),S_i\big)
    -
    \widehat{\Esp}\!\left[\delta\!\big(\pi(\Xvec_i),S\big)| \pi=\pi(\Xvec_i)\right].
\end{align*}
This centering guarantees that the average correction at each premium level coincides with the marginal balance adjustment while redistributing residual bias across values of $S$.
By construction, $ \widehat\pi_{\mathrm{mbc}}(\Xvec) \approx \Esp\!\left[Y | \pi(\Xvec),S\right]$,
while preserving approximate autocalibration in the premium dimension.

This procedure constitutes the continuous analogue of the group-wise multibalance correction described in Section~\ref{subsec:iso_mbc}. While it enforces conditional balance in a single smoothing step, its performance depends on the local density of observations in the joint $(\pi,S)$ space. In regions with limited exposure, two-dimensional smoothing may become unstable, which motivates the regulariaed iterative alternative introduced next.

\subsubsection{Regularized empirical enforcement of multicalibration via iterative bias correction}
\label{subsec:iter_cont}

We introduce a regularized empirical procedure that enforces multicalibration sequentially. Rather than estimating the full conditional mean surface in one step, we iteratively remove residual bias in the joint $(\pi,S)$ space, combining one-dimensional and two-dimensional smoothers with exposure-driven shrinkage to stabilize local corrections.

\paragraph{Local residual biases.}

Let $\pi^{(0)}(\Xvec)=\pi(\Xvec)$ denote the baseline premium. For iterations $j=0,1,2,\ldots,J$, define the residual
\begin{align*}
    R_i^{(j)} = Y_i - \pi^{(j)}(\Xvec_i).
\end{align*}
We estimate two smooth conditional mean functions at each iteration:
\begin{align*}
    \widehat b^{(j)}(p) = \Esp\!\left[Y-\pi^{(j)}(\Xvec) | \pi^{(j)}(\Xvec)=p\right],
\end{align*}
\begin{align*}
    \widehat b^{(j)}(p,s) = \Esp\!\left[Y-\pi^{(j)}(\Xvec) | \pi^{(j)}(\Xvec)=p,\; S=s\right].
\end{align*}
The first term corresponds to the marginal premium-level bias and plays the role of the bin-level bias in the discrete algorithm, thereby enforcing autocalibration with respect to the premium. The second captures systematic deviations that remain conditional on both the premium level and the grouping variable. The one-dimensional and two-dimensional smooth regressions are performed in the current premium space $\pi^{(j)}(\Xvec)$.

\paragraph{Exposure-driven shrinkage.}

To stabilize local bias estimates in regions with limited data, we introduce credibility weights in the joint $(\pi,S)$ space. Prior to the iterative procedure, we use a $k$-nearest-neighbor estimator to approximate the local density of observations around each point $(p,s)$ and convert it into a local effective exposure $w_{\mathrm{loc}}(p,s)$. We define the credibility weight as
\begin{align*}
    z(p,s) = \frac{w_{\mathrm{loc}}(p,s)}{w_{\mathrm{loc}}(p,s)+c},
\end{align*}
where $c>0$ is a tuning parameter controlling the amount of shrinkage.

In contrast with the categorical case, where credibility weights depend on iteration-specific cell exposures, the quantity $w_{\mathrm{loc}}(p,s)$ reflects intrinsic data availability in the continuous $(\pi,S)$ space. We therefore compute the weights $z(p,s)$ once and keep them fixed throughout the iterations. This choice ensures that shrinkage captures structural sparsity of the data rather than fluctuations induced by successive premium updates and avoids feedback between correction and regularization.

We define the credibility-weighted conditional correction pointwise as
\begin{align*}
    \delta^{(j)}(p,s) = z(p,s)\Big(\widehat b^{(j)}(p,s) - \widehat b^{(j)}(p)\Big).
\end{align*}
As in Section~\ref{subsec:biv_mbc}, we apply a centering step to preserve autocalibration with respect to the premium level. We therefore define
\begin{align*}
    \widetilde b^{(j)}(p,s) = \widehat b^{(j)}(p) + \delta^{(j)}(p,s) - \Esp\!\left[\delta^{(j)}(p,S)| \pi^{(j)}(\Xvec)=p\right].
\end{align*}
We compute the centering term via a one-dimensional local regression of $\delta^{(j)}(\pi^{(j)}(\Xvec),S)$ on $\pi^{(j)}(\Xvec)$.
\begin{align*}
    \widetilde b^{(j)}(p,s) = \widehat b^{(j)}(p) + \delta^{(j)}(p,s) - \Esp\!\left[\delta^{(j)}(p,S)| \pi^{(j)}(\Xvec)=p\right].
\end{align*}
We estimate the conditional expectation in the last term using the same one-dimensional local smoothing step in the premium score. This guarantees that the correction integrates to $\widehat b^{(j)}(p)$ at each premium level, thereby maintaining marginal autocalibration while redistributing residual bias across values of the continuous grouping variable $S$.

\paragraph{Update step.}

For a step size $\eta\in(0,1]$, we update the premium additively as
\begin{align*}
    \pi^{(j+1)}(\Xvec_i) = \pi^{(j)}(\Xvec_i) + \eta\,\widetilde b^{(j)}\!\big(\pi^{(j)}(\Xvec_i),S_i\big).
\end{align*}
The step-size parameter $\eta$ controls the speed of adjustment, and we choose it to ensure stable convergence of the algorithm.

\paragraph{Stopping criterion.}

At initialization, we construct fixed grids for the premium and the grouping variable using quantile-based partitions of $\pi^{(0)}(\Xvec)$ and $S$. We use these grids exclusively to evaluate convergence. For each cell $(k,\ell)$ of this fixed grid, we compute the exposure-weighted averages $\overline\pi^{(j)}_{k\ell}$ and $\overline b^{(j)}_{k\ell}$ of the current premium and correction. We stop the procedure when
\begin{align*}
    \max_{k,\ell} \frac{\left|\eta\,\overline b^{(j)}_{k\ell}\right|}{\overline\pi^{(j)}_{k\ell}} \le \delta,
\end{align*}
for a user-specified tolerance parameter $\delta>0$.

\paragraph{Interpretation.}

Upon completion of the algorithm, we obtain a premium $\pi^{(J)}(\cdot)$ for which no systematic residual bias remains at the chosen resolution in the joint $(\pi,S)$ space. 

\paragraph{On the role of credibility weights.}

As in the categorical case, we interpret the link with classical B\"uhlmann--Straub credibility in an algebraic rather than structural sense. We do not postulate any latent parameter $\Theta$, and the procedure remains purely cross-sectional and algorithmic.

In the continuous setting, we view
\begin{align*}
    b(p,s) = \Esp\!\left[Y-\pi(\Xvec)| \pi(\Xvec)=p,\; S=s\right]
\end{align*}
as a latent systematic bias surface attached to the current premium. Observations located near $(p,s)$ provide repeated information for estimating this conditional mean residual. When local effective exposure $w_{\mathrm{loc}}(p,s)$ is small, these estimates become unstable, which motivates shrinkage toward the marginal bias.

The credibility weight
$  z(p,s) = \frac{w_{\mathrm{loc}}(p,s)}{w_{\mathrm{loc}}(p,s)+c}$
therefore acts as an exposure-driven regularization device, reducing variance in sparse regions while allowing more granular adjustments where sufficient data are available.
We thus use credibility as a stabilization mechanism rather than as the consequence of a hierarchical stochastic model, and the multicalibration procedure remains a post-processing algorithm enforcing conditional balance.

\subsubsection{Numerical illustrations}

We now illustrate the continuous multicalibration procedures using the same MTPL frequency dataset and data split as in Section~\ref{subsec:illustrdisc}. Here, we consider vehicle age as a continuous grouping variable and set
\begin{align*}
    S = \texttt{VehAge}.
\end{align*}
Vehicle age is not used as a rating variable but only as a calibration dimension.
We obtain baseline premiums $\pi(\cdot)$ from the same Poisson GAM with log-link and exposure offset described previously. These premiums exhibit systematic residual distortions across premium levels and conditional on vehicle age.

Starting from $\pi(\cdot)$, we consider five premium constructions. The first is the unawareness premium itself. The second corresponds to autocalibration obtained as a limiting case of the iterative algorithm when the credibility parameter $c$ is large, enforcing autocalibration
through one-dimensional local regression in the premium score. As a direct non-iterative benchmark targeting the same objective, we also consider balance correction defined by
\begin{align*}
    \pi_{\mathrm{bc}}(\Xvec)=\Esp[Y | \pi(\Xvec)],
\end{align*}
estimated via a one-dimensional local GLM.

We then compare two approaches enforcing conditional balance with respect to the continuous variable $S=\texttt{VehAge}$.
We obtain full multicalibration through the iterative procedure of Section~\ref{subsec:iter_cont}. At each iteration, we estimate the marginal bias $\widehat b^{(j)}(p)$ and the joint bias surface $\widehat b^{(j)}(p,s)$ via local GLMs in the current premium space, apply exposure-driven shrinkage, and centre the correction as described in Section~\ref{subsec:biv_mbc}. 
The local GLMs are estimated using the \texttt{locfit} function in \texttt{R} with neighborhood parameter $\alpha=0.5$, representing the nearest neighbor fraction. We use local linear fits, step size $\eta=0.2$, and stopping threshold $\delta=0.01$.

As a direct non-iterative alternative, we implement multibalance correction via the centered bivariate local GLM of Section~\ref{subsec:biv_mbc}. Specifically, we estimate the conditional mean surface
$  m(p,s)=\Esp[Y | \pi(\Xvec)=p,\; S=s]$
using a two-dimensional local Poisson regression in the joint $(\pi,S)$ space with the same smoothing parameters. We then apply the centering step to preserve marginal autocalibration.

We evaluate residual bias on the validation set using ten premium bins constructed from $\pi(\cdot)$ and for each value of \texttt{VehAge}. For visual clarity, these curves are summarised by a shaded envelope corresponding to the minimum and maximum residual bias across vehicle ages within each premium bin.
Figure~\ref{fig:MCvsMBCcont} displays the residual structure. The baseline premiums show a clear monotone distortion across premium bins, with substantial under-pricing in the highest bins. After autocalibration (iterative or direct balance correction), the global residual curve becomes essentially flat, confirming that both approaches enforce autocalibration, as shown on the top panel. However, dispersion across vehicle age remains visible within each premium bin, indicating persistent conditional bias with respect to $S$, as illustrated by the wider shaded areas in the middle panel of Figure~\ref{fig:MCvsMBCcont} for both approaches.

\begin{figure}[h!]
    \centering
    \includegraphics[width=\textwidth]{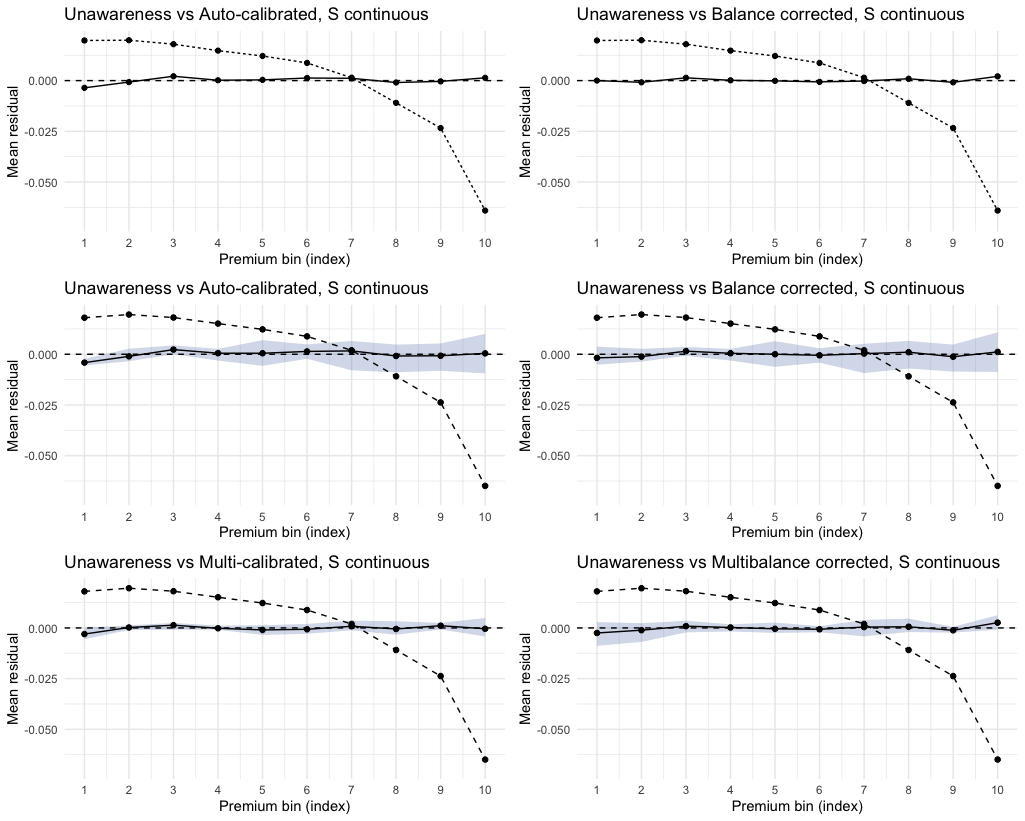}
    \caption{Mean residual pricing bias across premium bins for continuous vehicle-age. Top row: baseline vs.\ autocalibration. Middle row: autocalibration via iterative algorithm (left) and balance correction via isotonic regression (right). Bottom row: multicalibration (iterative) and multibalance correction (groupwise isotonic regression).}
    \label{fig:MCvsMBCcont}
\end{figure}

Applying full multicalibration or multibalance correction substantially reduces this dispersion. The shaded regions in the bottom row narrow considerably relative to autocalibration, demonstrating that both methods effectively remove conditional distortions along the continuous vehicle-age dimension. The iterative multicalibration procedure yields the smoothest residual patterns, reflecting the stabilizing effect of shrinkage and repeated bias updates. The direct multibalance correction achieves comparable alignment in a single step, with minor local variations attributable to finite-sample smoothing in two dimensions.

Table~\ref{tab:perf_test_cont} reports out-of-sample performance on the test set using Poisson deviance and Gini coefficient. The pattern is again consistent with the categorical case discussed in Section~\ref{subsec:illustrdisc}. The baseline model exhibits the highest deviance, reflecting score-level miscalibration. Autocalibration substantially reduces deviance, confirming that correcting premium-level distortions improves model adequacy.

Introducing conditional calibration with respect to continuous vehicle age yields further improvements. The multicalibrated premium achieves the lowest deviance overall, improving upon the autocalibrated specification. The non-iterative balance and multibalance corrections attain very similar but slightly higher deviance values. Overall, the iterative procedures (autocalibration and multicalibration) provide the best adequacy performance in the continuous setting.

In terms of discriminatory power, all calibrated models markedly improve upon the baseline. The Gini coefficients of the calibrated specifications are relatively close. Multibalance correction attains the highest Gini value, followed closely by balance correction and autocalibration, while multicalibration yields a slightly lower value. 

\begin{table}[htbp]
\centering
\caption{Out-of-sample performance comparison on the test set.}
\label{tab:perf_test_cont}
\begin{tabular}{lcc}
\hline
Model & Deviance & Gini Coefficient \\
\hline
Baseline & 36\,310.68 & 0.5783  \\
Autocalibrated & 33\,793.14 & 0.6941  \\
Balance correction & 33\,799.32  & 0.6973  \\
Multicalibrated & \textbf{33\,781.86} & 0.6904 \\
Multibalance correction & 33\,814.61 & \textbf{0.7023}  \\
\hline
\end{tabular}
\end{table}

\color{black}

\section{Discussion}

This paper implemented multicalibration for insurance premiums. This notion is particularly appealing
because it ensures financial equilibrium and guarantees that the expected return of the
insurance operation is identical across groups defined on the basis of the sensitive feature.
Several methods for implementing multicalibration have been proposed and illustrated on a motor
insurance portfolio. Interestingly, multicalibration improves the performances of the candidate
premium whereas imposing fairness generally deteriorates accuracy.

\section*{Acknowledgements}

The first Author
gratefully acknowledges funding from the FWO and F.R.S.-FNRS
under the Excellence of Science (EOS) programme, project ASTeRISK (40007517).
The third Author gratefully acknowledges the support received from the Research Chair ACTIONS under the aegis of the Risk Foundation, an initiative by BNP Paribas Cardif and the French Institute of Actuaries.


\section*{References}

\begin{enumerate}
\item[-]
Barocas, S., Hardt, M., Narayanan, A. (2019).
Fairness and Machine Learning. 
fairmlbook.org.
\item[-]
Baumann, J., Loi, M. (2023). 
Fairness and risk: An ethical argument for a group fairness definition insurers can use. 
Philosophy \& Technology 36, article \# 45.
\item[-]
Charpentier, A. (2024).
Insurance, Biases, Discrimination and Fairness.
Springer.
\item[-]
Ciatto, N., Verelst, H., Trufin, J.,  Denuit, M. (2023).
Does autocalibration improve goodness of lift?
European Actuarial Journal 13, 479-486.
\item[-]
Denuit, M., Charpentier, A., Trufin, J. (2021).
Autocalibration and Tweedie-dominance for insurance pricing with machine learning.
Insurance: Mathematics and Economics 101, 485-497.
\item[-]
Denuit, M., Dhaene, J., Goovaerts, M.J., Kaas, R. (2005).
Actuarial Theory for Dependent Risks: Measures, Orders and Models. 
Wiley, New York.
\item[-]
Denuit, M., Sznajder, D., Trufin, J. (2019).
Model selection based on Lorenz and concentration curves, Gini indices and convex order.
Insurance: Mathematics and Economics 89, 128-139.

\item[-]
Denuit, M., Trufin, J. (2023).
Model selection with Pearson's correlation, concentration and Lorenz curves under autocalibration.
European Actuarial Journal 13, 871-878.
\item[-]
Denuit, M., Trufin, J. (2024).
Convex and Lorenz orders under balance correction in nonlife insurance pricing:
Review and new developments.
Insurance: Mathematics and Economics 118, 123-128.
\item[-]
Dwork, C., Hardt, M., Pitassi, T., Reingold, O., Zemel, R. (2012). 
Fairness through awareness. 
Proceedings Innovations in Theoretical Computer Science Conference (ITCS 2012), pp. 214–226.
\item[-]
Fissler, T., Lorentzen, C., Mayer, M. (2022). 
Model comparison and calibration assessment: User guide for consistent scoring functions in machine learning and actuarial practice. 
arXiv preprint arXiv:2202.12780.
\item[-]
Frees E. W., Huang F. (2023).
The discriminating (pricing) actuary.
North American Actuarial Journal 27, 2-24.
\item[-]
Fr\"ohlich, C., Williamson, R.C. (2024). 
Insights from insurance for fair machine learning.
Proceedings of the 2024 ACM Conference on Fairness, Accountability, and Transparency, pp. 407-421.
\item[-]
Hebert-Johnson, U., Kim, M., Reingold, O., Rothblum, G. (2018). 
Multicalibration: Calibration for the (computationally-identifiable) masses.
Proceedings of International Conference on Machine Learning, pp. 1939-1948.
\item[-]
Kr\"uger, F., Ziegel, J.F. (2021).
Generic conditions for forecast dominance.
Journal of Business \& Economic Statistics 39, 972-983.
\item[-]
Kusner, M. J., Loftus, J., Russell, C., Silva, R. (2017).
Counterfactual Fairness. In Guyon, I., Luxburg,
U. V., Bengio, S., Wallach, H., Fergus, R., Vishwanathan, S., Garnett, R., (Eds.), Advances in
Neural Information Processing Systems, vol. 30.
\item[-]
Lindholm, M., Lindskog, F., Palmquist, J. (2023). 
Local bias adjustment, duration-weighted probabilities, and automatic construction of tariff cells. 
Scandinavian Actuarial Journal 2023, 946-973.
\item[-]
Lindholm, M., Richman, R., Tsanakas A., Wüthrich, M.V. (2022). 
Discrimination-free insurance pricing.
Astin Bulletin 52, 55-89.
\item[-]
Lundborg, A. R. (2022). 
Modern methods for variable significance testing.
PhD thesis, Apollo - University of Cambridge Repository.
\item[-]
Noll, A., Salzmann, R., W\"uthrich, M. (2018). Case study: French motor third-party liability
claims. Available at SSRN https://ssrn.com/abstract=3164764.
\item[-]
Scantamburlo, T., Baumann, J., Heitz, C. (2025). 
On prediction-modelers and decision-makers: Why fairness requires more than a fair prediction model. 
AI \& Society 40, 353-369.
\item[-]
W\"uthrich, M. V. (2023). 
Model selection with Gini indices under auto-calibration. 
European Actuarial Journal 13, 469-477.

\item[-]
W\"uthrich, M. V., Ziegel, J. (2024). 
Isotonic recalibration under a low signal-to-noise ratio.
Scandinavian Actuarial Journal 2024(3), 279-299

\item[-]
Zhang, Y., Huang, L., Yang, Y., Shao, X. (2025).
Testing conditional mean independence using generative neural networks. 
arXiv preprint arXiv:2501.17345.

\end{enumerate}

\end{document}